\documentclass[american,aps,prl,reprint]{revtex4-2}
\usepackage[T1]{fontenc}
\usepackage[latin9]{inputenc}
\usepackage{color}
\usepackage{babel}
\usepackage{amsmath}
\usepackage{amsthm}
\usepackage{amssymb}
\usepackage[unicode=true,pdfusetitle,
 bookmarks=true,bookmarksnumbered=false,bookmarksopen=false,
 breaklinks=false,pdfborder={0 0 0},pdfborderstyle={},backref=false,colorlinks=true]
 {hyperref}
\hypersetup{
 allcolors=magenta}

\makeatletter
\theoremstyle{plain}
\newtheorem{thm}{\protect\theoremname}
\theoremstyle{plain}
\newtheorem{prop}[thm]{\protect\propositionname}
\ifx\proof\undefined
\newenvironment{proof}[1][\protect\proofname]{\par
	\normalfont\topsep6\p@\@plus6\p@\relax
	\trivlist
	\itemindent\parindent
	\item[\hskip\labelsep\scshape #1]\ignorespaces
}{%
	\endtrivlist\@endpefalse
}
\providecommand{\proofname}{Proof}
\fi

\usepackage{times}
\usepackage{txfonts}
\usepackage{braket}
\usepackage{colortbl}

\makeatother

\providecommand{\theoremname}{Theorem}
\providecommand{\propositionname}{Proposition}

\begin{document}
\title{Coherence generation with Hamiltonians}

\author{Manfredi Scalici}

\author{Moein Naseri}

\author{Alexander Streltsov}
\affiliation{Centre for Quantum Optical Technologies, Centre of New Technologies,
University of Warsaw, Banacha 2c, 02-097 Warsaw, Poland}

\begin{abstract}
    We explore methods to generate quantum coherence through unitary evolutions, by introducing and studying the coherence generating capacity of Hamiltonians. This quantity is defined as the maximum derivative of coherence that can be achieved by a Hamiltonian. By adopting the relative entropy of coherence as our figure of merit, we evaluate the maximal coherence generating capacity with the constraint of a bounded Hilbert-Schmidt norm for the Hamiltonian. Our investigation yields closed-form expressions for both Hamiltonians and quantum states that induce the maximal derivative of coherence under these conditions. Specifically, for qubit systems, we solve this problem comprehensively for any given Hamiltonian, identifying the quantum states that lead to the largest coherence derivative induced by the Hamiltonian. Our investigation enables a precise identification of conditions under which quantum coherence is optimally enhanced, offering valuable insights for the manipulation and control of quantum coherence in quantum systems.
\end{abstract}

\maketitle

\section{Introduction}

The paradigm of quantum resource theories~\cite{RevModPhys.91.025001} provides a structured approach to investigate the characteristics of quantum systems and their utility in quantum technologies. Notably, the resource theories of entanglement~\cite{RevModPhys.81.865} and coherence~\cite{RevModPhys.89.041003,Wu2021} stand out as significant examples within this framework. The resource theory of entanglement explores the capabilities and constraints of spatially separated agents who operate within their local quantum laboratories and communicate via classical channels~\cite{RevModPhys.81.865}. On the other hand, the resource theory of coherence delves into the challenges and opportunities faced by an agent who is limited in the abilities to generate and preserve quantum coherence~\cite{RevModPhys.89.041003,Wu2021}. Furthermore, the framework of quantum resource theories has been adeptly applied to the field of quantum thermodynamics~\cite{Horodecki2013,Ng2018}. This application has enabled a deeper understanding of how quantum systems can be manipulated within the bounds of energy constraints. 

Every quantum resource theory is grounded on defining two basic elements: free states and free operations~\cite{RevModPhys.91.025001}. Free states refer to quantum states that can be easily generated within a setting that is justified by physical principles. Specifically, within the resource theory of coherence, free states are identified as incoherent states~\cite{PhysRevLett.113.140401}. These are states which are diagonal in a certain reference basis, denoted by $\ket{i}$. The impetus behind exploring this theory is rooted in the phenomenon of unavoidable decoherence, suggesting that incoherent states are those that remain unchanged in the presence of decoherence.

On the aspect of free operations, these ideally are quantum manipulations that can be easily executed, based on the physical considerations on which the resource theory is based. For instance, in the resource theory of entanglement, the set of local operations and classical communication embodies a set of free operations endowed with tangible physical interpretation~\cite{PhysRevA.54.3824}. Within the context of the resource theory of coherence, various sets of free operations have been scrutinized~\cite{PhysRevLett.113.140401,WinterPhysRevLett.116.120404,aberg2006quantifying,YadinPhysRevX.6.041028,ChitambarPhysRevLett.117.030401,Gour_2008,deVicente_2017}. A common characteristic among these sets is their inability to generate coherence from incoherent states. 

At the heart of any quantum resource theory lies the fundamental inquiry into the feasibility of state transformations. It is customary to consider the scenario where $n$ copies of a given initial state $\rho$ are at one's disposal, with the ambition to transform these into $m$ copies of a desired target state $\sigma$. This process is envisioned to accommodate an error margin that diminishes as the number of initial state copies, $n$, increases. The efficiency of this transformation is quantified by the highest feasible ratio of $m/n$, signifying the transformation rate.

Within the framework of the resource theory of coherence, the optimal rate at which one quantum state can be transformed into another is precisely determined when focusing on the set of maximally incoherent operations (MIO). These operations are characterized by their inability to produce coherence from states that are initially incoherent~\cite{aberg2006quantifying}. The maximal transformation rate in this setting is encapsulated by the formula~\cite{WinterPhysRevLett.116.120404}:
\begin{equation}
R(\rho\rightarrow\sigma)=\frac{C_\mathrm{r}(\rho)}{C_\mathrm{r}(\sigma)}, \label{eq:MIOrate}
\end{equation}
where $C_\mathrm{r}(\rho)$ represent the relative entropy of coherence for the initial and target states, respectively. The relative entropy of coherence for a state $\rho$ is defined as~\cite{PhysRevLett.113.140401}:
\begin{equation}
C_\mathrm{r}(\rho)=\min_{\sigma\in\mathcal{I}}S(\rho||\sigma)=S(\Delta[\rho])-S(\rho),
\end{equation}
with the quantum relative entropy $S(\rho||\sigma)=\mathrm{Tr}[\rho\log_{2}\rho]-\mathrm{Tr}[\rho\log_{2}\sigma]$, the von Neumann entropy $S(\rho)=-\mathrm{Tr}[\rho\log_{2}\rho]$, and $\Delta[\rho]=\sum_{i}\ket{i}\!\bra{i}\rho\ket{i}\!\bra{i}$ representing the operation of complete dephasing in the incoherent basis. A similar quantity has been studied previously also within the resource theory of entanglement~\cite{PhysRevLett.78.2275}.

The relative entropy of coherence is an important example of a coherence measure~\cite{PhysRevLett.113.140401}. Essentially, a coherence measure, represented by $C(\rho)$, quantifies the amount of coherence in a quantum state $\rho$. Its fundamental characteristic is that it is non-increasing under the application of free operations, denoted as $\Lambda_f$, that are permissible within the resource theory framework, i.e.,  $C(\Lambda_f[\rho]) \leq C(\rho)$ for any quantum state $\rho$ and any free operation $\Lambda_f$. The literature offers a diverse range of coherence quantifiers, each grounded in either physical principles or mathematical foundations~\cite{aberg2006quantifying,PhysRevLett.113.140401,PhysRevLett.115.020403,PhysRevA.94.060302,PhysRevA.93.032326,PhysRevA.92.022124,deVicente_2017}.

Considering the significance of quantum coherence in quantum information science and quantum technology~\cite{RevModPhys.89.041003,Wu2021}, it becomes essential to explore and comprehend the most efficient methodologies for its generation. One strategy to create quantum coherence involves employing a static quantum channel, represented by $\Lambda$. This approach can successfully create coherence from an initially incoherent state, contingent upon the condition that $\Lambda$ is not in the set MIO. The pursuit of identifying and refining optimal methods for the generation of coherence through static quantum channels has garnered attention and been subject to detailed examination in various studies~\cite{PhysRevLett.113.140401,PhysRevA.92.032331,Diaz2018usingreusing,PhysRevA.105.L060401}.

In this article, we focus on the optimal methods to generate coherence via dynamical evolutions, focusing specifically on unitary evolutions $U_{t}=e^{-itH}$. Using the relative entropy of coherence as a figure of merit, we investigate the maximal derivative of $C_\mathrm{r}$ achievable in this setting, maximized over all initial states $\rho$ and all Hamiltonians $H$ with bounded Hilbert-Schmidt norm. This quantity has a clear operational meaning via Eq.~(\ref{eq:MIOrate}), corresponding to the maximal coherence generation rate achievable via Hamiltonians with bounded Hilbert-Schmidt norm. We characterize optimal initial states and optimal Hamiltonians for any system of dimension $d$. For qubit systems, we provide optimal input state for any given Hamiltonian.

\section{Coherence generating capacity of Hamiltonians}

For a Hamiltonian $H$, we define the \emph{coherence generating capacity} of $H$ as the maximal increase of coherence achievable via a unitary evolution $U_{t}=e^{-itH}$ at time $t=0$, i.e.,
\begin{equation}
\left.C_{\mathrm{gen}}(H)=\max_{\rho}\frac{C_{\mathrm{r}}(e^{-iHt}\rho e^{iHt})}{dt}\right|_{t=0}.
\end{equation}
Functions of this form have been previously studied in entanglement theory, in the context of entanglement generation via non-local Hamiltonians~\cite{PhysRevLett.87.137901,PhysRevA.76.052319}.

The following proposition provides an alternative expression for the coherence generating capacity.

\begin{prop}
For any Hamiltonian $H$ it holds that
\begin{equation}
C_{\mathrm{gen}}(H)=\max_{\rho}i\mathrm{Tr}(H[\rho,\log_{2}\Delta(\rho)]).
\end{equation}
\end{prop}

\begin{proof}
    We define the state $\rho_{t}=e^{-iHt}\rho e^{iHt}$ and $\dot{\rho}_{t}=d\rho_{t}/dt$.
It then holds
\begin{equation}
\frac{d}{dt}S(\rho_{t})=-\mathrm{Tr}[\dot{\rho}_{t}\log_{2}\rho_{t}]. \label{eq:EntropyDerivative}
\end{equation}
A proof of this equality is given in the Appendix. For the time derivative of coherence we obtain 
\begin{align}
\frac{dC_{\mathrm{r}}(\rho_{t})}{dt} & =\frac{dS(\Delta[\rho_{t}])}{dt}=-\mathrm{Tr}\left[\left(\frac{d}{dt}\Delta(\rho_{t})\right)\log_{2}\Delta(\rho_{t})\right]\nonumber \\
 & =-\mathrm{Tr}\left[\Delta(\dot{\rho}_{t})\log_{2}\Delta(\rho_{t})\right],\label{eq:Cderivative}
\end{align}
where we have used the fact that dephasing commutes with the time
derivative, i.e., 
\begin{equation}
\frac{d}{dt}\Delta(\rho_{t})=\Delta(\dot{\rho}_{t}).
\end{equation}
Using the von Neumann equation $\dot{\rho}_{t}=-i[H,\rho_{t}]$ Eq.~(\ref{eq:Cderivative})
can be expressed as 
\begin{equation}
\frac{dC_{\mathrm{r}}(\rho_{t})}{dt}=i\mathrm{Tr}\left[\Delta([H,\rho_{t}])\log_{2}\Delta(\rho_{t})\right].
\end{equation}
As we further show in the Appendix,  
\begin{equation}
\mathrm{Tr}\left[\Delta(A)\log_{2}\Delta(B)\right]=\mathrm{Tr}\left[A\log_{2}\Delta(B)\right] \label{eq:DeltaLog}
\end{equation}
holds true for any Hermitian matrix $A$ and any positive matrix $B$. Choosing $A=i\Delta([H,\rho_{t}])$ and
$B=\rho_{t}$ we further obtain 
\begin{equation}
\frac{dC_{\mathrm{r}}(\rho_{t})}{dt}=i\mathrm{Tr}\left[[H,\rho_{t}]\log_{2}\Delta(\rho_{t})\right].
\end{equation}
At time $t=0$ we further have 
\begin{equation}
\left.\frac{dC_{\mathrm{r}}(\rho_{t})}{dt}\right|_{t=0}=i\mathrm{Tr}[[H,\rho]\log_{2}\Delta(\rho)],
\end{equation}
where $\rho=\rho_{t=0}$. This expression can be further written as
\begin{align}
\left.\frac{dC_\mathrm{r}(\rho_{t})}{dt}\right|_{t=0} & =i\mathrm{Tr}(H[\rho,\log_{2}\Delta(\rho)]) \label{eq:Cderivative-2}.
\end{align}
Performing the maximum over all states $\rho$ completes the proof.

\end{proof}

In the following, our goal is to evaluate the maximal coherence generating
capacity over all Hamiltonians with the constraint $||H||_{2}\leq1$, where $||M||_{2}=\sqrt{\mathrm{Tr}[M^{\dagger}M]}$ is the Hilbert-Schmidt norm of a matrix $M$.
As we will see, this problem is closely related to the variance of
the surprisal, as investigated in~\citep{7001656} (a similar technique has been used earlier in~\cite{PhysRevA.76.052319}). For
a probability distribution $\boldsymbol{p}=(p_{0},\ldots,p_{d-1})$,
the the surprisal $-\log_{2}p_{i}$ is a quantifier of the surprise
to obtain the outcome $i$. The variance of the surprisal is given
as 
\begin{equation}
f(\boldsymbol{p})=\sum_{i}p_{i}\left(-\log_{2}p_{i}\right)^{2}-\left[\sum_{i}p_{i}\left(-\log_{2}p_{i}\right)\right]^{2}.\label{eq:f}
\end{equation}
As we will show in the following theorem, the maximal coherence generating
capacity for a system of dimension $d$ is closely related to the
maximal variance of the surprisal $f$.
\begin{thm}
\label{thm:CoherenceGeneration}It holds that 
\begin{equation}
\max_{||H||_{2}\leq1}C_{\mathrm{gen}}(H)=\max_{\boldsymbol{p}}\sqrt{2f(\boldsymbol{p})}.
\end{equation}
\end{thm}
\begin{proof}
Defining the Hermitian matrix 
\begin{equation}
    M=i[\rho,\log_{2}\Delta(\rho)],
\end{equation}
the derivative of the relative entropy of coherence at time zero can be written as 
\begin{equation}
\left.\frac{dC_{\mathrm{r}}(\rho_{t})}{dt}\right|_{t=0}=\mathrm{Tr}(HM).
\end{equation}
We can now perform the maximization over all Hamiltonians
$H$ with the property $||H||_{2}\leq1$. For this, we use H\"older's
inequality, arriving at 
\begin{equation}
\left.\frac{dC_{\mathrm{r}}(\rho_{t})}{dt}\right|_{t=0}=\mathrm{Tr}(HM)\leq||H||_{2}||M||_{2}.
\end{equation}
For a given $M$, this inequality is saturated if $H$ is chosen as
\begin{equation}
H=\frac{M}{||M||_{2}}.\label{eq:Hopt}
\end{equation}
Maximizing over all Hamiltonians with bounded Hilbert-Schmidt norm, we obtain
\begin{equation}
\max_{||H||_{2}\leq1}\left.\frac{dC_{\mathrm{r}}(\rho_{t})}{dt}\right|_{t=0}=\frac{\mathrm{Tr}[M^{2}]}{||M||_{2}}=||M||_{2}=\left\Vert [\rho,\log_{2}\Delta(\rho)]\right\Vert _{2}.
\end{equation}

To complete the proof of the theorem, it remains to maximize $\left\Vert [\rho,\log_{2}\Delta(\rho)]\right\Vert _{2}$
over all states $\rho$. Denoting with $\rho_{ij}$ the elements of
$\rho$, we obtain 
\begin{align} \label{eq:Commutator}
[\rho,\log_{2}\Delta(\rho)] & =\rho\log_{2}\Delta(\rho)-\left[\log_{2}\Delta(\rho)\right]\rho\\
 & =\sum_{i,j}\left(\rho_{ij}\log_{2}\rho_{jj}-\rho_{ij}\log_{2}\rho_{ii}\right)\ket{i}\!\bra{j}\nonumber \\
 & =\sum_{i,j}\rho_{ij}\left(\log_{2}\rho_{jj}-\log_{2}\rho_{ii}\right)\ket{i}\!\bra{j}.\nonumber 
\end{align}
With this, we obtain the following:
\begin{equation}
\left\Vert [\rho,\log_{2}\Delta(\rho)]\right\Vert _{2}^{2}=\sum_{i,j}|\rho_{ij}|^{2}\left(\log_{2}\rho_{jj}-\log_{2}\rho_{ii}\right)^{2}.
\end{equation}
Since $\rho$ is a quantum state, it holds that 
\begin{equation}
\rho_{ii}\rho_{jj}\geq|\rho_{ij}|^{2},
\end{equation}
which implies the inequality 
\begin{equation}
\left\Vert [\rho,\log_{2}\Delta(\rho)]\right\Vert _{2}^{2}\leq\sum_{i,j}\rho_{ii}\rho_{jj}\left(\log_{2}\rho_{jj}-\log_{2}\rho_{ii}\right)^{2}.\label{eq:Proof-1}
\end{equation}

Let us now define the function
\begin{equation}
f(\rho)=\frac{1}{2}\sum_{i,j}\rho_{ii}\rho_{jj}\left(\log_{2}\rho_{jj}-\log_{2}\rho_{ii}\right)^{2}.\label{eq:f-1}
\end{equation}
Note that the right-hand side of Eq.~(\ref{eq:Proof-1}) corresponds
to $2f(\rho)$. As we prove in the Appendix, $f(\rho)$
can also be written as
\begin{equation}
f(\rho)=\sum_{i}\rho_{ii}\left(-\log_{2}\rho_{ii}\right)^{2}-\left[\sum_{i}\rho_{ii}\left(-\log_{2}\rho_{ii}\right)\right]^{2}.\label{eq:f-2}
\end{equation}
Note that this function coincides with the variance of the surprisal
function defined in Eq.~(\ref{eq:f}), if we choose $p_{i}=\rho_{ii}$.

Consider now a pure state of the form
\begin{equation}
\ket{\psi}=\sum_{i=0}^{d-1}\sqrt{q_{i}}\ket{i},\label{eq:psi}
\end{equation}
where the probabilities $q_{i}$ are chosen such that the variance
of the surprisal is maximal. Consider now the density matrix
$\sigma=\ket{\psi}\!\bra{\psi}$. Due to the arguments presented above,
it is clear that $\sigma$ maximizes the function $f$, i.e., $f(\sigma)=\max_{\rho}f(\rho)$.
Moreover, the matrix elements of $\sigma$ fulfill $\sigma_{ii}\sigma_{jj}=|\sigma_{ij}|^{2}$,
which implies 
\begin{align}
2f(\sigma) & =\sum_{i,j}\sigma_{ii}\sigma_{jj}\left(\log_{2}\sigma_{jj}-\log_{2}\sigma_{ii}\right)^{2}\\
 & =\sum_{i,j}|\sigma_{ij}|^{2}\left(\log_{2}\sigma_{jj}-\log_{2}\sigma_{ii}\right)^{2}\nonumber \\
 & =\left\Vert [\sigma,\log_{2}\Delta(\sigma)]\right\Vert _{2}^{2}.\nonumber 
\end{align}

Collecting the arguments presented above, we have 
\begin{align}
\max_{\rho}\left\Vert [\rho,\log_{2}\Delta(\rho)]\right\Vert _{2}^{2} & \leq\max_{\rho}2f(\rho)=2f(\sigma)\\
 & =\left\Vert [\sigma,\log_{2}\Delta(\sigma)]\right\Vert _{2}^{2}\nonumber \\
 & \leq\max_{\rho}\left\Vert [\rho,\log_{2}\Delta(\rho)]\right\Vert _{2}^{2}.\nonumber 
\end{align}
This proves that 
\begin{equation}
\max_{\rho}\left\Vert [\rho,\log_{2}\Delta(\rho)]\right\Vert _{2}^{2}=2f(\sigma),
\end{equation}
and the proof of the theorem is complete.
\end{proof}
An optimal initial state for coherence generation can be given as
follows:
\begin{align}
\ket{\psi} & =\sqrt{\gamma}\ket{0}+\sqrt{\frac{1-\gamma}{d-1}}\sum_{i=1}^{d-1}\ket{i},
\end{align}
where $\gamma\in(0,1)$ is chosen such that the probability distribution
$(\gamma,\frac{1-\gamma}{d-1},\dots,\frac{1-\gamma}{d-1})$ maximizes
the variance of the surprisal~\citep{7001656,PhysRevA.76.052319}. Following
the arguments from the proof of Theorem~\ref{thm:CoherenceGeneration},
an optimal Hamiltonian is given by Eq.~(\ref{eq:Hopt}) with 
\begin{align}
M & =i[\psi,\log_{2}\Delta(\psi)]=i\sum_{k,l}\psi_{kl}\left(\log_{2}\psi_{ll}-\log_{2}\psi_{kk}\right)\ket{k}\!\bra{l}\nonumber \\
 & =i\sqrt{\gamma}\sqrt{\frac{1-\gamma}{d-1}}\left(\log_{2}\frac{1-\gamma}{d-1}-\log_{2}\gamma\right)\sum_{l=1}^{d-1}\ket{0}\!\bra{l}\nonumber \\
 & +i\sqrt{\gamma}\sqrt{\frac{1-\gamma}{d-1}}\left(\log_{2}\gamma-\log_{2}\frac{1-\gamma}{d-1}\right)\sum_{k=1}^{d-1}\ket{k}\!\bra{0}\nonumber \\
 & =i\alpha\left(\ket{0}\!\bra{\phi}-\ket{\phi}\!\bra{0}\right)
\end{align}
with a state $\ket{\phi}=\sum_{i=1}^{d-1}\ket{i}/\sqrt{d-1}$ and
some $\alpha\in\mathbb{R}$. An optimal Hamiltonian can thus be chosen
as 
\begin{equation}
H=\frac{i}{\sqrt{2}}\left(\ket{0}\!\bra{\phi}-\ket{\phi}\!\bra{0}\right).
\end{equation}

We will now focus explicitly on the single-qubit case. In this case, we will evaluate $C_{\mathrm{gen}}(H)$ for any Hamiltonian $H$. In the following, we denote the elements of the density matrix with $\rho_{kl}$, and similarly $H_{kl}$ are elements of $H$. Moreover, $\rho_{01}=|\rho_{01}|e^{i\alpha}$ and similarly $H_{01}=|H_{01}|e^{i\beta}$. Using Eq.~(\ref{eq:Commutator}) we obtain 
\begin{align}
i\mathrm{Tr}\left(H[\rho,\log_{2}\Delta(\rho)]\right) & =i\sum_{k,l}H_{lk}\rho_{kl}\left(\log_{2}\rho_{ll}-\log_{2}\rho_{kk}\right)\\
 & =i\left[H_{10}\rho_{01}\left(\log_{2}\rho_{11}-\log_{2}\rho_{00}\right)\right]\nonumber \\
 & +i\left[H_{01}\rho_{10}\left(\log_{2}\rho_{00}-\log_{2}\rho_{11}\right)\right]\nonumber \\
 & =i\left[H_{10}\rho_{01}-H_{01}\rho_{10}\right]\log_{2}\frac{\rho_{11}}{\rho_{00}}\nonumber \\
 & =i\left|H_{10}\right|\left|\rho_{01}\right|\left[e^{i(\alpha-\beta)}-e^{-i(\alpha-\beta)}\right]\log_{2}\frac{\rho_{11}}{\rho_{00}}\nonumber \\
 & =-2\left|H_{10}\right|\left|\rho_{01}\right|\sin(\alpha-\beta)\log_{2}\frac{\rho_{11}}{\rho_{00}}.\nonumber 
\end{align}
Our goal now is to maximize this expression over all $\alpha$, $|\rho_{01}|$, $\rho_{00}$, and  $\rho_{11}$, taking into account that $\rho$ is a density matrix of a single qubit. Maximizing over $\alpha$ is straightforward, an optimal choice is $\alpha = \beta - \pi/2$. We thus arrive at 
\begin{equation}
C_{\mathrm{gen}}(H)=\max_{\rho_{ij}}2\left|H_{10}\right|\left|\rho_{01}\right|\log_{2}\frac{\rho_{11}}{\rho_{00}}.
\end{equation}
For any qubit density matrix it holds that $|\rho_{01}|\leq \sqrt{\rho_{00} \rho_{11}}$ with equality on pure states. With this we can perform the maximization over $|\rho_{01}|$, leading to 
\begin{equation}
C_{\mathrm{gen}}(H)=\max_{\rho_{ij}}2\left|H_{10}\right|\sqrt{\rho_{00}\rho_{11}}\log_{2}\frac{\rho_{11}}{\rho_{00}}.
\end{equation}
This also means that an optimal state can be chosen to be pure. In the last step we recall that $\rho_{11} = 1-\rho_{00}$, such that 
\begin{equation}
C_{\mathrm{gen}}(H)=\max_{\rho_{00}}2\left|H_{10}\right|\sqrt{\rho_{00}(1-\rho_{00})}\log_{2}\frac{1-\rho_{00}}{\rho_{00}}.
\end{equation}
This maximization can be performed numerically, leading to $\rho_{00}\approx 0.083$. 

We note that analogous results have been previously reported in entanglement theory. In particular, optimal entanglement generation for two-qubit Hamiltonians has been considered in~\cite{PhysRevLett.87.137901}, and optimal states for entanglement generation without ancillas have been derived. Optimal entanglement generation with Hamiltonians of bounded operator norm have been investigated in~\cite{PhysRevA.76.052319}.

\section{Conclusions}

In conclusion, our investigation offers an in-depth exploration of the coherence generating capacity inherent to Hamiltonians, showcasing a methodology that enables the assessment of the maximum coherence derivative achievable in quantum systems of any dimension, governed by Hamiltonians with bounded Hilbert-Schmidt norms. Specifically, for qubit systems, we have achieved a comprehensive resolution of this problem for any Hamiltonian, identifying states that maximize the rate of change in the relative entropy of coherence.

This inquiry opens the door to several compelling questions for future research. A primary area of interest is the extent of coherence enhancement attainable through specific Hamiltonians in systems of dimensions greater than those of qubits. Although our approach provides a novel way to frame this question, the feasibility of solving this maximization problem analytically, or possibly through semidefinite programming, remains to be determined. Additionally, the potential applicability of our techniques to other quantum resource theories, especially the resource theory of entanglement, poses an intriguing prospect. Considering the parallels between the resource theories of coherence and entanglement, there is a promising possibility that our strategies could uncover optimal methods for increasing entanglement in a system using certain Hamiltonian classes.

\section*{Acknowledgements}
We thank Marek Miller for discussion. This work was supported by the ``Quantum Coherence and Entanglement for Quantum Technology'' project, carried out within the First Team programme of the Foundation for Polish Science co-financed by the European Union under the European Regional Development Fund, and the National Science Centre Poland (Grant No. 2022/46/E/ST2/00115).

\bibliography{literature}
\clearpage

\onecolumngrid

\appendix

\section{Proof of Eq.~(\ref{eq:EntropyDerivative})}
\label{AppendixA}
Here we will prove that the time derivative of the von Neumann entropy can be written as
\begin{equation}
\frac{d}{dt}S(\rho_{t})=-\mathrm{Tr}\left[\dot{\rho}_{t}\log_2\rho_{t}\right].
\end{equation}
First, we decompose the density matrix $\rho_t$ in its eigenbasis $\rho_{t}=\sum_{i}\lambda_{i}\ket{\psi_{i}}\!\bra{\psi_{i}}$, where $\lambda_i$ and $\ket{\psi_i}$ are time-dependent eigenvalues and eigenstates, respectively. It follows that:
\begin{align}
\frac{d}{dt}\left(\rho_{t}\log_{2}\rho_{t}\right) & =\sum_{i}\frac{d}{dt}\left(\lambda_{i}\frac{\ln\lambda_{i}}{\ln2}\right)\ket{\psi_{i}}\!\bra{\psi_{i}}+\sum_{i}\left(\lambda_{i}\log_{2}\lambda_{i}\right)\frac{d}{dt}\ket{\psi_{i}}\!\bra{\psi_{i}}\\
 & =\sum_{i}\left(\dot{\lambda}_{i}\frac{\ln\lambda_{i}}{\ln2}+\frac{\dot{\lambda}_{i}}{\ln2}\right)\ket{\psi_{i}}\!\bra{\psi_{i}}+\sum_{i}\left(\lambda_{i}\log_{2}\lambda_{i}\right)\left[\dot{\ket{\psi_{i}}}\!\bra{\psi_{i}}+\ket{\psi_{i}}\!\dot{\bra{\psi_{i}}}\right].\nonumber 
\end{align}
This means that the derivative of the von Neumann entropy can be written as
\begin{align}
\frac{d}{dt}S(\rho_{t}) & = - \mathrm{Tr}\left[\frac{d}{dt}\left(\rho_{t}\log_{2}\rho_{t}\right)\right]= - \sum_{i}\left(\dot{\lambda}_{i}\frac{\ln\lambda_{i}}{\ln2}+\frac{\dot{\lambda}_{i}}{\ln2}\right),
\end{align}
where we have used the fact that 
\begin{equation}
\mathrm{Tr}\left[\dot{\ket{\psi_{i}}}\!\bra{\psi_{i}}+\ket{\psi_{i}}\!\dot{\bra{\psi_{i}}}\right]=\mathrm{Tr}\left[\frac{d}{dt}\left(\ket{\psi_{i}}\!\bra{\psi_{i}}\right)\right]=\frac{d}{dt}\mathrm{Tr}\left[\ket{\psi_{i}}\!\bra{\psi_{i}}\right]=0.
\end{equation}
Noting that $\sum_{i}\dot{\lambda}_{i}=\frac{d}{dt}\sum_{i}\lambda_{i}=0$, we obtain 
\begin{equation}
\frac{d}{dt}S(\rho_{t})= - \sum_{i}\dot{\lambda}_{i}\log_{2}\lambda_{i}.
\end{equation}

In the next step, we write $\dot{\rho}\log_{2}\rho$ as follows:
\begin{equation}
\dot{\rho}\log_{2}\rho=\sum_{i,j}\left(\dot{\lambda}_{i}\ket{\psi_{i}}\!\bra{\psi_{i}}+\lambda_{i}\frac{d}{dt}\ket{\psi_{i}}\!\bra{\psi_{i}}\right)\log_{2}\lambda_{j}\ket{\psi_{j}}\!\bra{\psi_{j}},
\end{equation}
which implies
\begin{align}
\mathrm{Tr}\left[\dot{\rho}\log_{2}\rho\right] & =\sum_{i,j}\left(\dot{\lambda}_{i}\log_{2}\lambda_{j}\mathrm{Tr}\left[\ket{\psi_{i}}\!\braket{\psi_{i}|\psi_{j}}\!\bra{\psi_{j}}\right]+\lambda_{i}\log_{2}\lambda_{j}\mathrm{Tr}\left[\frac{d}{dt}\left(\ket{\psi_{i}}\!\bra{\psi_{i}}\right)\ket{\psi_{j}}\!\bra{\psi_{j}}\right]\right)\\
 & =\sum_{i}\dot{\lambda}_{i}\log_{2}\lambda_{i}+\sum_{i,j}\lambda_{i}\log_{2}\lambda_{j}\mathrm{Tr}\left[\left(\dot{\ket{\psi_{i}}}\!\bra{\psi_{i}}+\ket{\psi_{i}}\!\dot{\bra{\psi_{i}}}\right)\ket{\psi_{j}}\!\bra{\psi_{j}}\right]\nonumber \\
 & =\sum_{i}\dot{\lambda}_{i}\log_{2}\lambda_{i}+\sum_{i}\lambda_{i}\log_{2}\lambda_{i}\left(\braket{\psi_{i}|\dot{\psi}_{i}}+\braket{\dot{\psi}_{i}|\psi_{i}}\right)=\sum_{i}\dot{\lambda}_{i}\log_{2}\lambda_{i}.\nonumber 
\end{align}
This completes the proof.

\section{Proof of Eq.~(\ref{eq:DeltaLog})}
Let $A$ be a Hermitian matrix, and $B$ be a positive matrix. It
holds that 
\begin{align}
\mathrm{Tr}\left[A\log_{2}\Delta(B)\right] & =\mathrm{Tr}\left[A\log_{2}\left(\sum_{i}\ket{i}\!\bra{i}B\ket{i}\!\bra{i}\right)\right]=\sum_{i}\mathrm{Tr}\left[A\log_{2}\left(\ket{i}\!\bra{i}B\ket{i}\!\bra{i}\right)\right]\nonumber \\
 & =\sum_{i,j}\mathrm{Tr}\left[A\ket{j}\!\bra{j}\log_{2}\left(\ket{i}\!\bra{i}B\ket{i}\!\bra{i}\right)\ket{j}\!\bra{j}\right]\nonumber \\
 & =\sum_{i,j}\mathrm{Tr}\left[\ket{j}\!\bra{j}A\ket{j}\!\bra{j}\log_{2}\left(\ket{i}\!\bra{i}B\ket{i}\!\bra{i}\right)\right]\nonumber \\
 & =\mathrm{Tr}\left[\sum_{j}\ket{j}\!\bra{j}A\ket{j}\!\bra{j}\log_{2}\left(\sum_{i}\ket{i}\!\bra{i}B\ket{i}\!\bra{i}\right)\right]\nonumber \\
 & =\mathrm{Tr}\left[\Delta(A)\log_{2}\Delta(B)\right].
\end{align}

\section{Proof of Eq.~(\ref{eq:f-2})}

Here we will show that the function 
\begin{equation}
f(\rho)=\frac{1}{2}\sum_{i,j}\rho_{ii}\rho_{jj}\left(\log_{2}\rho_{jj}-\log_{2}\rho_{ii}\right)^{2}
\end{equation}
can also be written as 
\begin{equation}
f(\rho)=\sum_{i}\rho_{ii}\left(-\log_{2}\rho_{ii}\right)^{2}-\left[\sum_{i}\rho_{ii}\left(-\log_{2}\rho_{ii}\right)\right]^{2}.
\end{equation}
This follows directly from the following equalities:

\begin{align}
\sum_{i}\rho_{ii}\left(-\log_{2}\rho_{ii}\right)^{2}-\left[\sum_{i}\rho_{ii}\left(-\log_{2}\rho_{ii}\right)\right]^{2} & =\sum_{i}\rho_{ii}\left(\log_{2}\rho_{ii}\right)^{2}-\sum_{i}\rho_{ii}^{2}\left(\log_{2}\rho_{ii}\right)^{2}-\sum_{i\neq j}\rho_{ii}\rho_{jj}\log_{2}\rho_{ii}\log_{2}\rho_{jj}\\
 & =\sum_{i}\rho_{ii}\left(1-\rho_{ii}\right)\left(\log_{2}\rho_{ii}\right)^{2}-\sum_{i\neq j}\rho_{ii}\rho_{jj}\log_{2}\rho_{ii}\log_{2}\rho_{jj}\nonumber \\
 & =\sum_{i}\rho_{ii}\left(\sum_{j\neq i}\rho_{jj}\right)\left(\log_{2}\rho_{ii}\right)^{2}-\sum_{i\neq j}\rho_{ii}\rho_{jj}\log_{2}\rho_{ii}\log_{2}\rho_{jj}\nonumber \\
 & =\sum_{i\neq j}\rho_{ii}\rho_{jj}\left(\log_{2}\rho_{ii}\right)^{2}-\sum_{i\neq j}\rho_{ii}\rho_{jj}\log_{2}\rho_{ii}\log_{2}\rho_{jj}\nonumber \\
 & =\frac{1}{2}\sum_{i\neq j}\rho_{ii}\rho_{jj}\left(\log_{2}\rho_{ii}\right)^{2}+\frac{1}{2}\sum_{i\neq j}\rho_{ii}\rho_{jj}\left(\log_{2}\rho_{jj}\right)^{2}-\sum_{i\neq j}\rho_{ii}\rho_{jj}\log_{2}\rho_{ii}\log_{2}\rho_{jj}\nonumber \\
 & =\frac{1}{2}\sum_{i\neq j}\rho_{ii}\rho_{jj}\left[\left(\log_{2}\rho_{ii}\right)^{2}+\left(\log_{2}\rho_{jj}\right)^{2}-2\log_{2}\rho_{ii}\log_{2}\rho_{jj}\right]\nonumber \\
 & =\frac{1}{2}\sum_{i\neq j}\rho_{ii}\rho_{jj}\left(\log_{2}\rho_{ii}-\log_{2}\rho_{ii}\right)^{2}=\frac{1}{2}\sum_{i,j}\rho_{ii}\rho_{jj}\left(\log_{2}\rho_{ii}-\log_{2}\rho_{ii}\right)^{2}.\nonumber 
\end{align}
Here, we have used the fact that $\sum_{j\neq i}\rho_{jj}=1-\rho_{ii}$.
\end{document}